\documentclass{article}
\usepackage[utf8]{inputenc}
\usepackage{amsmath, amsthm, amssymb, verbatim, tikz,algorithm,algorithmic,placeins}
\usetikzlibrary{patterns}
\title{On the Complexity of Role Colouring Planar Graphs, Trees and Cographs}
\author{Christopher Purcell \& M. Puck Rombach\footnote{Department of Mathematics, University of California Los Angeles, 520 Portola Plaza, CA 90095, USA. {\tt rombach@math.ucla.edu}, +1 (310) 206-9974.}}
\date{February 2014}

\newtheorem{theorem}{Theorem}
\newtheorem{lemma}{Lemma}
\newtheorem{claim}{Claim}

\begin{document}

\maketitle

\section*{Abstract}
We prove several results about the complexity of the role colouring problem. A {\em role colouring} of a graph $G$ is an assignment of colours to the vertices of $G$ such that two vertices of the same colour have identical sets of colours in their neighbourhoods. We show that the problem of finding a role colouring with $1< k <n$ colours is NP-hard for planar graphs. We show that restricting the problem to trees yields a polynomially solvable case, as long as $k$ is either constant or has a constant difference with $n$, the number of vertices in the tree. Finally, we prove that cographs are always $k$-role-colourable for $1<k\leq n$ and construct such a colouring in polynomial time.

\section{Introduction}

A {\em role colouring} of a graph $G$ is an assignment of colours to the vertices of $G$ such that two vertices of the same colour have identical sets of colours in their neighbourhoods. The concept arises from the study of social networks. Network science is an increasingly important application of graph theory and role colourings are a natural formulation of roles played by nodes in a real-world network~\cite{rossi2014role,sailer1979structural}. This structure was formalised by White and Reitz in terms of graph homomorphisms in \cite{white1983graph}, and developed extensively by Borgatti and Everett~\cite{borgatti1993two,borgatti1992notions,everett1994regular,everett1991role}. A fast, applicable algorithm for finding role colourings is proposed in~\cite{hummell1987strukturbeschreibung,burt1990detecting}. A homomorphism $h$ is said to be {\em locally surjective} if $h$ is surjective when restricted to the neighbourhood set of any vertex. Locally surjective homomorphisms are equivalent to role colourings and they appear in the literature under many other names, \emph{e.g.} role assignment \cite{van2010computing}, role equivalence \cite{burt1990detecting}, regular equivalence \cite{borgatti1993two}. Throughout this paper we use the language of graph colourings and we refer to a role colouring using $k$ colours as a $k$-role-colouring. 

We consider the computational problem associated with role colourings whose input is a graph $G$ and whose output is a partition of the vertices of $G$ into $k$ non empty subsets satisfying the definition of a role colouring given above. We call this problem $k$-{\sc role-colourability}, or $k$-{\sc rolecol} for short. This problem differs from the more common {\sc colourability} problem in a few important ways. Firstly, every graph with no isolated vertex has a role colouring obtained by giving each vertex the same colour and every graph has a role colouring obtained by giving each vertex its own colour. Secondly a $k$-role-colouring does not usually imply the existence of a $k+1$-role-colouring. This makes makes role colouring inherently different from the original graph colouring problem.

Finding role colourings of a given size is known to be NP-complete in general \cite{roberts2001how, fiala2005complete}. For $k\geq 3$, the $k$-{\sc rolecol} problem is NP-complete when restricted to chordal graphs \cite{van2010computing}. However, 2-role-colouring can be solved in polynomial time for chordal graphs \cite{sheng2003two}. Not many other partial results on complexity of role colouring are known. In fact, interval graphs and trees are the only non trivial classes in which a polynomial solution is known to exist, and only for a constant number of colours.

The rest of this paper is organised as follows. In section~\ref{sec:planar}, we prove that $k$-{\sc rolecol} remains NP-complete even when restricted to planar graphs, a class that was suggested for examination in \cite{van2010computing} and is one of the most extensively studied in the literature. In Section~\ref{sec:trees}, we give an explicit algorithm that computes a $k$-role-colouring of a tree in polynomial time, as long as $k$ is either constant or has a constant difference with $n$, the number of vertices in $T$. Finally, in Section~\ref{sec:cographs}, we show that every cograph (with at least $k$ vertices) has a $k$-role-colouring, and hence that the decision version of the problem is solvable in polynomial time in this class. Our proof is constructive and gives an explicit algorithm to construct such a colouring.


\section{Planar Graphs}
\label{sec:planar}

In order to prove that $k$-{\sc rolecol} is NP-complete when restricted to planar graphs, we introduce the {\sc satisfiability problem}, defined below.
A boolean formula $\phi$ (in {\em conjunctive normal form}) is a set of clauses $C_1,C_2,\ldots$, each of which is a set of variables $x_1,x_2,\ldots$. The variables may take values TRUE or FALSE. For a given assignment of these values to the variables, a clause is said to be {\em satisfied} if at least one of its variables is assigned the value TRUE. A formula is satisfied if each of its clauses is satisfied. The {\sc satisfiability} problem takes a boolean formula on $n$ variables as its input and asks if there is an assignment of TRUE and FALSE to the variables that satisifies the formula. The general {\sc satisfiability} problem was the first to be revealed to be NP-complete \cite{cook1971complexity}, and remains a central problem in theoretical computer science. 

We will use a reduction from a certain restricted version of {\sc satisfiability}. In order to describe this restricted problem, we define the following graph theoretic notion. The {\em formula graph} $G_\phi$ of a given formula $\phi$ is a bipartite graph whose vertices correspond to the clauses and variables of $\phi$ with an edge between $C$ and $x$ if the variable $x$ appears in the clause $C$. Let {\sc $k$-satisfiability} be the {\sc satisfiability} problem with the restriction that each clause contains at most $k$ variables. The {\sc $3$-satisfiability} problem is NP-complete even when restricted to formulas with planar formula graphs~\cite{lichtenstein1982planar}. In \cite{tovey1984simplified}, Tovey showed that the problem is NP-complete under the restriction that each clause has two or three variables and each variable appears at most three times
We call the corresponding problem {\sc $3*,3*$-satisfiability}. We now combine the restrictions imposed by Tovey and planarity to show that {\sc planar $3*,3*$-satisfiability}, which is {\sc $3*,3*$-satisfiability} with planar formula graphs, is also NP-complete. We list a couple of planarity preserving operations that we will need throughout the coming proofs, in an easy lemma. See also~\cite{gross2005graph}.

\begin{lemma}\label{planarop}
If $G'$ is a graph created from a planar graph $G$ by any of the following operations, then $G'$ is planar.
\begin{description}
\item[(a)] Adding a path $x,z_1,\ldots,z_k,y$ where $x,y \in V(G)$, and $x,y$ share a face in some planar drawing of $G$, and $z_1,\ldots,z_k$ are new vertices.
\item[(b)] Replacing a vertex $x \in V(G)$ with $d_G(x)=k$ by a cycle $z_1,\ldots,z_k$ with edges $z_i,z_{i+1}$, $1 \leq i \leq k-1$ and $z_k,z_1$, and edges $z_i,y_i$, $1 \leq i \leq k$, where $y_1,\ldots,y_k$ are the neigbours in $G$ of $x$ appearing in clockwise order in a planar drawing of $G$.
\item[(c)] Attaching a new planar subgraph $H$ to $G$, such that $V(G')=V(G) \cup V(H)$, $E(G')=E(G) \cup E(H) \cup xz$, where $x \in V(G)$, $z \in V(H)$. 
\end{description}

\end{lemma}
\begin{proof}
\begin{description}
\item[(a)] Replacing an edge by a multi-edge does not destroy planarity and replacing an edge $xy$ by a path $x,z_1,\ldots,z_k,y$ clearly does not destroy planarity either.
\item[(b)] Cycles are planar and since the neighbours $y_1,\ldots,y_k$ and new vertices $z_1,\ldots,z_k$ are in the same order clockwise, the edges $z_i,y_i$, $1 \leq i \leq k$ do not cross each other or any new cycle-edges.
\item[(c)] We take the disjoint union of $G$ and $H$ and draw $G$ such that $x$ is on the outer face and $H$ such that $z$ is on the outer face. Adding an edge between two vertices on the same face does not destroy planarity.
\end{description}
\end{proof}

\begin{lemma}
The {\sc planar $3,3$-satisfiability} problem is NP-complete.
\end{lemma}

\begin{proof}
We follow the method~\cite{tovey1984simplified} of reducing any {\sc $3$-satisfiability} problem to a {\sc $3*,3*$-satisfiability} problem. The method is as follows. Suppose that variable $x$ appears in $k$ clauses. Create $k$ new variables $x_1,\ldots,x_k$ and replace the $i$th occurence of $x$ with $x_i$, for $1 \leq i \leq k$. Append the clause $\{ x_i \lor \bar{x}_{i+1} \}$ for $1 \leq i \leq k-1$ and $\{ x_k \lor \bar{x}_1 \}$. This new set of clauses forces the variables $x_i$, $1 \leq i \leq k$, to be either all true or all false. 

Let $\phi$ be a {\sc planar $3$-satisfiability} problem, with a given planar drawing of $G_\phi$. Suppose that variable $x$ appears in $k$ clauses. Create $k$ new variables $x_1,\ldots,x_k$ and replace vertex $x$ in $G_\phi$ by the cycle on new vertices $x_1,\ldots,x_k$. Replace the $j_i$th occurence of $x$ with $x_i$, for $1 \leq i \leq k$. Such that the clauses $C_{j_1},\ldots,C_{j_k}$ appear in clockwise order in the neighbourhood of $x$ in our given planar drawing of $G_\phi$ (planarity-preserving operation (b) in lemma \ref{planarop}). Finally, replace each edge $x_ix_{i+1}$ by a path $x_i, C_{x_i}, x_{i+1}$, for $1 \leq i \leq k-1$, and edge $x_kx_{1}$ by path $x_k, C_{x_k}, x_{1}$ (planarity-preserving operation (a) in lemma \ref{planarop}). The clause vertices $C_{x_i}$ represent the new clauses $\{ x_i \lor \bar{x}_{i+1} \}$.  

\end{proof}

\begin{theorem}\label{th:rc}
$k$-{\sc rolecol}, $k \geq 2$, is NP-hard for connected planar graphs.
\end{theorem}

\begin{proof}

Let $\phi$ be a planar boolean formula on $n$ variables $x_1,x_2,\ldots,x_n$ having $m$ clauses $C_1,C_2,\ldots,C_m$, such that each variable appears at most three times and each clause is of size 2 or 3. Let $G_\phi$ be its formula graph. We will construct a related planar graph $G'_\phi$ that has a $k$-role-colouring if and only if $\phi$ is satisfiable.
We split the proof into two cases.

First suppose $k=2$. We construct $G'_\phi$ from $G_\phi$ as follows. To each clause vertex $C_j$ we add a path $a_j,b_j$ with edge $b_jC_j$ (lemma \ref{planarop} (c)). Since a variable appears in at most three clauses, at most one of $x_i$ and $\bar{x_i}$ appears twice. Furthermore, we can assume that a variable appears both positively and negatively. Subsequently, one of $x_i$ and $\bar{x_i}$ appears exactly once. If $\bar{x_i}$ appears exactly once, in clause $C_j$, we replace the edge $x_i C_j$ by a path $x_i,\bar{x_i},C_j$. Otherwise, we have that $x_i$ appears exactly once, in which case we relabel the node $x_i$ to $\bar{x_i}$, and we replace the edge $\bar{x_i} C_j$ by a path $\bar{x_i},x_i,C_j$ (lemma \ref{planarop} (a)). Finally, we add a vertex $y_i$ to each pair $x_i$ and $\bar{x_i}$ to form a triangle (lemma \ref{planarop} (a)). See figure \ref{fig:2col}.

For each variable $x_i$, the graph $G'_\phi$ will contain a copy of a triangle. With a slight abuse of notation for the sake of readability, we label these vertices $x_i$,$\bar{x_i}$,$y_i$. For each clause $C_j$, $G'_\phi$ contains a path on three vertices labelled $a_j,b_j,C_j$. Again this slight abuse of notation will aid the reader, and we refer to the vertices $x_1,\bar{x_1},\ldots,x_n,\bar{x_n}$ as {\em literal vertices} and the vertices $C_1,\ldots,C_m$ as {\em clause vertices}. We add an edge between a literal vertex $x$ and a clause vertex $C$ if the literal $x$ appears in the clause $C$.

Suppose that $G'_\phi$ has a 2-role-colouring. Without loss of generality, the vertex $a_1$ is red. It is easy to see that $b_1$ cannot also be red, for otherwise red vertices could only have red neighbours, and every vertex would then be red which would be a contradiction. Suppose $C_1$ is red. Then, since $b_1$ is blue and has only red neighbours, the neighbourhood of every blue vertex must be red. Let $x$ be a neighbour of $C_1$ other than $b_1$. Since $x$ is a neighbour of a red vertex it must be coloured blue. But $x$ is contained in a triangle, the other two vertices of which must be coloured red. This contradiction shows that $C_1$ must be coloured blue, and we can deduce that each of the paths $a_j,b_j,C_j$ must be coloured red,blue,blue from the fact that the vertex $a_j$ has degree 1 and since a blue vertex must have at least one neighbour of each colour, $a_j$ must be coloured red.

\FloatBarrier
\begin{figure}
\begin{center}
\begin{tikzpicture}[thick,scale=0.7, every node/.style={scale=0.8}]
\draw[black!50,line width=.8pt] (0,0) -- (1,0);
\draw[black!50,line width=.8pt] (0,0) -- (.5,-.8);
\draw[black!50,line width=.8pt] (1,0) -- (.5,-.8);
\draw[black!50,line width=.8pt] (6,0) -- (7,0);
\draw[black!50,line width=.8pt] (6,0) -- (6.5,-.8);
\draw[black!50,line width=.8pt] (7,0) -- (6.5,-.8);
\draw[black!50,line width=.8pt] (.5,4) -- (0,0);
\draw[black!50,line width=.8pt] (.5,4) -- (.5,6);
\draw[black!50,line width=.8pt] (.5,4) -- (6,0);
\draw[black!50,line width=.8pt] (6.5,4) -- (6.5,6);
\draw[black!50,line width=.8pt] (6.5,4) -- (7,0);
\filldraw[fill=white, draw=black,line width=1pt] (0,0) circle (.25);
\draw[fill=black!100,black!100] (1,0) circle (.25);
\draw[fill=black!100,black!100] (.5,-.8) circle (.25);
\filldraw[fill=white, draw=black,line width=1pt] (7,0) circle (.25);
\draw[fill=black!100,black!100] (6,0) circle (.25);
\draw[fill=black!100,black!100] (6.5,-.8) circle (.25);
\draw[fill=black!100,black!100] (.5,4) circle (.25);
\draw[fill=black!100,black!100] (.5,5) circle (.25);
\filldraw[fill=white, draw=black,line width=1pt] (.5,6) circle (.25);
\draw[fill=black!100,black!100] (6.5,4) circle (.25);
\draw[fill=black!100,black!100] (6.5,5) circle (.25);
\filldraw[fill=white, draw=black,line width=1pt] (6.5,6) circle (.25);
\draw (-.6,0) node{$x_1$};
\draw (1.6,0) node{$\bar{x_1}$};
\draw (5.4,0) node{$x_2$};
\draw (7.6,0) node{$\bar{x_2}$};
\draw (.5,-1.3) node{$y_1$};
\draw (6.5,-1.3) node{$y_2$};
\draw (-.1,4) node{$C_1$};
\draw (5.9,4) node{$C_2$};
\draw (0,5) node{$b_1$};
\draw (6,5) node{$b_2$};
\draw (0,6) node{$a_1$};
\draw (6,6) node{$a_2$};
\end{tikzpicture}
\caption{A graph $G_\phi$ representing the boolean formula $\phi = \{x_1 \cup x_2 \} \cap \{ \bar{x_2} \}$ with a 2-role colouring corresponding to a satisfying assignment where $x_1$ and $\bar{x_2}$ are true.}\label{fig:2col}
\end{center}
\end{figure}
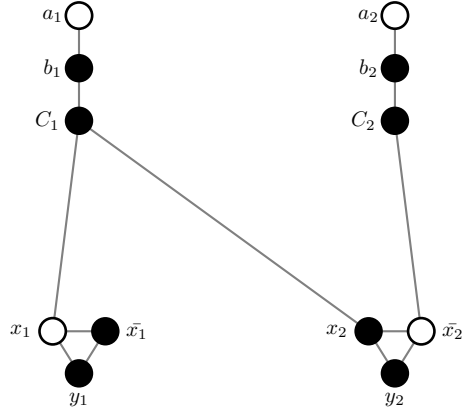
\begin{figure}
\begin{center}
\begin{tikzpicture}[thick,scale=0.7, every node/.style={scale=0.8}]
\node (x1) at (0,0){};
\draw[black!50,line width=.8pt] (0,0) -- (.45,0);
\draw[black!50,line width=.8pt,dashed] (.95,0) -- (3.05,0);
\draw[black!50,line width=.8pt] (3.55,0) -- (3.75,0);
\draw[black!50,line width=.8pt] (4.25,0) -- (4.7,0);
\draw[black!50,line width=.8pt] (0,0) -- (1.7,-1.5);
\draw[black!50,line width=.8pt] (4.7,0) -- (1.7,-1.5);
\draw[black!50,line width=.8pt] (1.7,-1.5) -- (1.7,-2.9);
\draw[black!50,line width=.8pt,dashed] (1.7,-4.2) -- (1.7,-2.9);
\draw[fill=black!100,black!100] (4.7,0) circle (.25);
\draw[pattern=horizontal lines,line width=1pt] (.7,0) circle (.25);
\filldraw[fill=black!20, draw=black,line width=1pt] (2,0) circle (.25);
\draw[pattern=vertical lines,line width=1pt] (3.3,0) circle (.25);
\draw[fill=black!100,black!100] (4,0) circle (.25);
\filldraw[fill=white, draw=black,line width=1pt] (1.7,-1.5) circle (.25);
\draw[pattern=vertical lines,line width=1pt] (1.7,-1.5) circle (.25);
\filldraw[fill=white, draw=black,line width=1pt] (1.7,-2.2) circle (.25);
\filldraw[fill=white, draw=black,line width=1pt] (1.7,-2.9) circle (.25);
\draw[pattern=horizontal lines,line width=1pt] (1.7,-2.9) circle (.25);
\filldraw[fill=black!20, draw=black,line width=1pt] (1.7,-4.2) circle (.25);
\draw (-.6,0) node{$x_1$};
\draw (5.3,0) node{$\bar{x_1}$};
\draw (.7,.5) node{$z_{1,1}$};
\draw (1.8,.5) node{$z_{1,k-3}$};
\draw (3.2,.5) node{$z_{1,2k-5}$};
\draw (4,.8) node{$z_{1,2k-4}$};
\draw (2.6,-1.5) node{$y_{1,k-1}$};
\draw (2.6,-2.2) node{$y_{1,k-2}$};
\draw (2.6,-2.9) node{$y_{1,k-3}$};
\draw (2.4,-4.2) node{$y_{1,1}$};
\pgftransformxshift{300}
\node (x2bar) at (1.3,0){};
\node (x2) at (-3.4,0){};
\draw[black!50,line width=.8pt] (1.3,0) -- (.85,0);
\draw[black!50,line width=.8pt,dashed] (.35,0) -- (-1.75,0);
\draw[black!50,line width=.8pt] (-2.25,0) -- (-2.45,0);
\draw[black!50,line width=.8pt] (-2.95,0) -- (-3.4,0);
\draw[black!50,line width=.8pt] (1.3,0) -- (-1.7,-1.5);
\draw[black!50,line width=.8pt] (-3.4,0) -- (-1.7,-1.5);
\draw[black!50,line width=.8pt] (-1.7,-1.5) -- (-1.7,-2.9);
\draw[black!50,line width=.8pt,dashed] (-1.7,-4.2) -- (-1.7,-2.9);
\draw[fill=black!100,black!100] (-3.4,0) circle (.25);
\filldraw[pattern=horizontal lines,line width=1pt] (.6,0) circle (.25);
\filldraw[fill=black!20, draw=black,line width=1pt] (-.7,0) circle (.25);
\draw[pattern=vertical lines,line width=1pt] (-2,0) circle (.25);
\draw[fill=black!100,black!100] (-2.7,0) circle (.25);
\filldraw[fill=white, draw=black,line width=1pt] (-1.7,-1.5) circle (.25);
\draw[pattern=vertical lines,line width=1pt] (-1.7,-1.5) circle (.25);
\filldraw[fill=white, draw=black,line width=1pt] (-1.7,-2.2) circle (.25);
\filldraw[fill=white, draw=black,line width=1pt] (-1.7,-2.9) circle (.25);
\draw[pattern=horizontal lines,line width=1pt] (-1.7,-2.9) circle (.25);
\filldraw[fill=black!20, draw=black,line width=1pt] (-1.7,-4.2) circle (.25);
\draw (1.9,0) node{$\bar{x_2}$};
\draw (-4,0) node{$x_2$};
\draw (-2.7,.5) node{$z_{2,1}$};
\draw (-2,.5) node{$z_{2,2}$};
\draw (-.8,.5) node{$z_{2,k}$};
\draw (.47,.5) node{$z_{2,2k-4}$};
\draw (-2.6,-1.5) node{$y_{2,k-1}$};
\draw (-2.6,-2.2) node{$y_{2,k-2}$};
\draw (-2.6,-2.9) node{$y_{2,k-3}$};
\draw (-2.4,-4.2) node{$y_{2,1}$};
\pgftransformxshift{-270}
\pgftransformyshift{100}
\draw[black!50,line width=.8pt,dashed] (1.7,3.6) -- (1.7,4.6);
\draw[black!50,line width=.8pt] (1.7,1.5) -- (1.7,3.6);
\draw[black!50,line width=.8pt] (1.7,1.5) -- (1.7,0);
\draw[black!50,line width=.8pt] (1.7,1.5) -- (1,0);
\draw[black!50,line width=.8pt] (1.7,1.5) -- (2.4,0);
\draw[black!50,line width=.8pt] (2.4,0) -- (1,0);
\draw[black!50,line width=.8pt] (node cs:name=x1) -- (1.7,0);
\filldraw[fill=white, draw=black,line width=1pt] (node cs:name=x1) circle (.25);
\draw[black!50,line width=.8pt] (node cs:name=x2) -- (1.7,0);
\filldraw[fill=white, draw=black,line width=1pt] (1.7,0) circle (.25);
\draw[pattern=vertical lines,line width=1pt] (1.7,0) circle (.25);
\draw[fill=black!100,black!100] (1,0) circle (.25);
\draw[fill=black!100,black!100] (2.4,0) circle (.25);
\draw[fill=black!100,black!100] (1.7,1.5) circle (.25);
\filldraw[fill=white, draw=black,line width=1pt] (1.7,2.2) circle (.25);
\draw[pattern=vertical lines,line width=1pt] (1.7,2.2) circle (.25);
\filldraw[fill=white, draw=black,line width=1pt] (1.7,2.9) circle (.25);
\filldraw[fill=white, draw=black,line width=1pt] (1.7,2.9) circle (.25);
\filldraw[fill=white, draw=black,line width=1pt] (1.7,3.6) circle (.25);
\draw[pattern=horizontal lines,line width=1pt] (1.7,3.6) circle (.25);
\filldraw[fill=black!20, draw=black,line width=1pt] (1.7,4.6) circle (.25);
\draw (1.9,-.6) node{$C_1$};
\draw (.3,0) node{$u_{1,1}$};
\draw (3.1,0) node{$u_{1,2}$};
\draw (2.4,1.5) node{$v_{1,k}$};
\draw (2.6,2.2) node{$v_{1,k-1}$};
\draw (2.6,2.9) node{$v_{1,k-2}$};
\draw (2.6,3.6) node{$v_{1,k-3}$};
\draw (2.4,4.6) node{$v_{1,1}$};
\pgftransformxshift{280}
\draw[black!50,line width=.8pt,dashed] (-1.7,3.6) -- (-1.7,4.6);
\draw[black!50,line width=.8pt] (-1.7,1.5) -- (-1.7,3.6);
\draw[black!50,line width=.8pt] (-1.7,1.5) -- (-1.7,0);
\draw[black!50,line width=.8pt] (-1.7,1.5) -- (-1,0);
\draw[black!50,line width=.8pt] (-1.7,1.5) -- (-2.4,0);
\draw[black!50,line width=.8pt] (-2.4,0) -- (-1,0);
\draw[black!50,line width=.8pt] (node cs:name=x2bar) -- (-1.7,0);
\filldraw[fill=white, draw=black,line width=1pt] (node cs:name=x2bar) circle (.25);
\filldraw[fill=white, draw=black,line width=1pt] (-1.7,0) circle (.25);
\draw[pattern=vertical lines,line width=1pt] (-1.7,0) circle (.25);
\draw[fill=black!100,black!100] (-1,0) circle (.25);
\draw[fill=black!100,black!100] (-2.4,0) circle (.25);
\draw[fill=black!100,black!100] (-1.7,1.5) circle (.25);
\filldraw[fill=white, draw=black,line width=1pt] (-1.7,2.2) circle (.25);
\draw[pattern=vertical lines,line width=1pt] (-1.7,2.2) circle (.25);
\filldraw[fill=white, draw=black,line width=1pt] (-1.7,2.9) circle (.25);
\filldraw[fill=white, draw=black,line width=1pt] (-1.7,2.9) circle (.25);
\filldraw[fill=white, draw=black,line width=1pt] (-1.7,3.6) circle (.25);
\draw[pattern=horizontal lines,line width=1pt] (-1.7,3.6) circle (.25);
\filldraw[fill=black!20, draw=black,line width=1pt] (-1.7,4.6) circle (.25);
\draw (-1.9,-.6) node{$C_2$};
\draw (-.3,0) node{$u_{2,2}$};
\draw (-3.1,0) node{$u_{2,1}$};
\draw (-2.4,1.5) node{$v_{2,k}$};
\draw (-2.6,2.2) node{$v_{2,k-1}$};
\draw (-2.6,2.9) node{$v_{2,k-2}$};
\draw (-2.6,3.6) node{$v_{2,k-3}$};
\draw (-2.4,4.6) node{$v_{2,1}$};
\end{tikzpicture}
\caption{A graph $G_\phi$ representing the boolean formula $\phi = \{x_1 \cup x_2 \} \cap \{ \bar{x_2} \}$ with a $k$-role colouring corresponding to a satisfying assignment where $x_1$ and $\bar{x_2}$ are true.}\label{fig:kcol}
\end{center}
\end{figure}
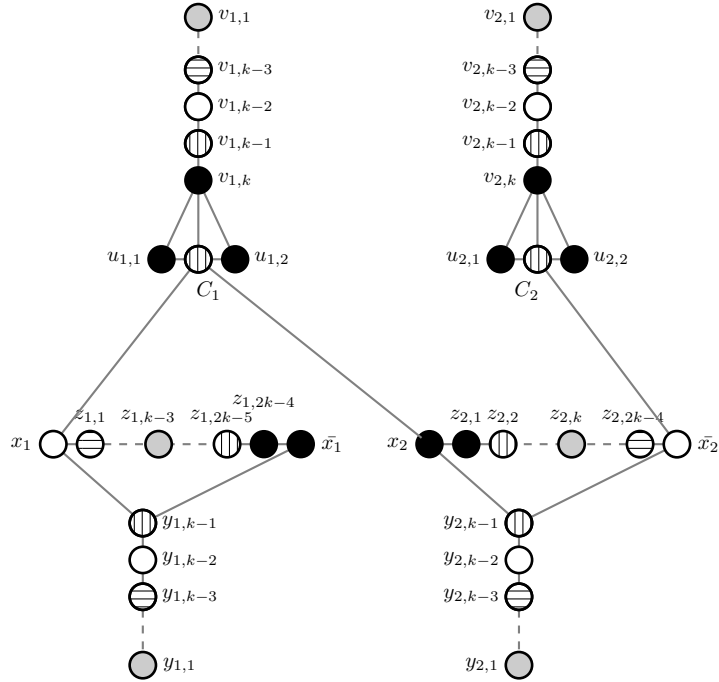
\FloatBarrier

Now consider the triangle $x_1,\bar{x_1},y_1$. They cannot all be coloured blue because blue vertices must have red neighbours. No two of them can be red, because red vertices cannot have red neighbours. Therefore exactly one of $x_1,\bar{x_1},y_1$ is coloured red. Observe that each clause vertex must have a red neighbour amongst the literal vertices in its neighbourhood. We may therefore construct a truth assignment by assigning TRUE to the variable $x_i$ if and only if the literal vertex $x_i$ is red.

In order to prove the result for $k>2$ we use a slightly different construction, and introduce
the following notation. An induced path $\{v_1,v_2,...,v_k\}$ in a graph
$G$ is said to be {\em dangling} if $v_1$ is of degree 1 in $G$ and 
$v_2,\ldots,v_{k-1}$ are of degree 2 in $G$.

\begin{lemma}
If a connected graph $G$ has a dangling path $P$ of length at most $k$, then in any $k$-role-colouring of $G$ no two vertices of $P$ can have the same colour
\end{lemma}

\begin{proof}
Let $\{v_1,v_2,\ldots,v_k\}$ be a dangling path in $P$ such that $v_1$ has degree 1. We have observed that the closed neighbourhood of a vertex cannot be monochromatic. Without loss of generality, $v_1$ has colour 1 and $v_2$ has colour 2. Suppose that $v_i$ has colour $i$ for $i<j$. Clearly $v_j$ cannot have colour $c<j-2$. Suppose $v_j$ has colour $j-2$ or $j-1$. Then no vertex of colour at least $j$ can have a neighbour of colour at most $j-1$. This contradicts the connectedness of $G$. Therefore without loss of generality, $v_j$ has colour $j$.
\end{proof}

We now give a description of $G'_\phi$ in the case that $k$ is at least 3. Let $\phi$ be a {\sc planar $3,3$-satisfiability} problem on $n$ variables $x_1,x_2,\ldots,x_n$ having $m$ clauses $C_1,C_2,\ldots,C_m$, and let $G_\phi$ be its formula graph. We construct a related planar graph $G'_\phi$ that has a $k$-role-colouring if and only if $\phi$ is satisfiable. 

We construct $G'_\phi$ from $G_\phi$ as follows. To each clause vertex $C_j$ we add a dangling path $v_{j,1},\ldots,v_{k,k}$, with edge $v_{k,k},C_j$ (lemma \ref{planarop} (c)). To each pair $v_{k,k},C_j$ we add two vertices $u_{j,1},u_{j,2}$ that both form a triangle with $v_{k,k},C_j$ (lemma \ref{planarop} (a)).  As before, one of $x_i$ and $\bar{x_i}$ appears exactly once. If $\bar{x_i}$ appears exactly once, in clause $C_j$, we replace the edge $x_i C_j$ by a path $x_i,z{j,1},\ldots ,z{j,2k-4},\bar{x_i},C_j$. Otherwise we have that $x_i$ appears exactly once, in which case we relabel the node $x_i$ to $\bar{x_i}$, and replace the edge $\bar{x_i} C_j$ by a path $\bar{x_i},z{j,1},\ldots ,z{j,2k-4},x_i,C_j$ (lemma \ref{planarop} (a)). We add a vertex $y_{i,k-1}$ and attach it to $x_i$ and $\bar{x_i}$ (lemma \ref{planarop} (a)). Finally, we add a dangling path $y_{i,1},\ldots,y{i,k-1}$ to the graph through the edge $y{i,k-1},y{i,k}$ (lemma \ref{planarop} (c)). See figure \ref{fig:kcol}.
 
As in the first part of the proof, suppose $G'_\phi$ has a $k$-role-colouring. Observe that the path formed by $v_{i,1},\ldots,v_{i,k}$ is dangling, and must therefore be coloured with $k$ distinct colours. Without loss of generality $v_{1,1}$ has colour 1, and indeed $v_{1,j}$ has colour $j$. Therefore the neighbours of any vertex of colour $j$ must have colour $j-1$ or $j+1$ for $1<j<k$. Consider the vertices $u_{1,1},u_{1,2}$. If either of them has colour $k-1$, then $C_i$ must have colour $k-2$. But $C_1$ is adjacent to $v_{1,k}$ which has colour $k$ which leads to a contradiction. So $u_{1,1}$ and $u_{1,2}$ must have colour $k$ and therefore $C_1$ has colour $k-1$. 

Since the role-graph of this colouring of $G'_\phi$ is a simple path on the colours $1,\ldots,k$ with $k$ having a self loop, all vertices of degree 1 must have colour 1. So for each clause $C_i$, we have that $v_{i,1}$ has colour 1, $v_{i,2}$ has colour 2, and so on. This implies that every subgraph induced on vertices $C_i,v_{i,1},\ldots,v_{i,k},u_{i,1},u_{i,2}$ has the same colouring as described above for $i=1$.

For each variable $x_i$, we have that $y_{i,1}$ has colour 1, $y_{i,2}$ has colour 2,..., $y_{i,k-1}$ has colour $k-1$. This implies that the vertices $x_i$ and $\bar{x}_i$ must have colours $k-2$ and $k$, or vice versa. If $x_i$ receives colour $k-2$ and $\bar{x}_i$ receives colour $k$, then the vertices $z_{i,1}$ through $z_{i,k-3}$ through $z_{i,2k-4}$ must receive colours $k-3$ through $1$ through $k$. Alternatively, if $x_i$ receives colour $k$ and $\bar{x}_i$ receives colour $k-2$, then the vertices $z_{i,1}$ through $z_{i,k-3}$ through $z_{i,2k-4}$ must receive colours $k$ through $1$ through $k-3$. 

Now we construct a satisfying assignment for $\phi$ from this colouring. If the vertex representing $x_i$ has colour $k-2$ we assign the variable $x_i$ the value TRUE. If the vertex representing $\bar{x_i}$ has colour $k-2$ we assign the variable $\bar{x}_i$ TRUE. 

Since each vertex representing a clause must have a neighbour of colour $k-2$, each clause now has a variable or its negation that has been assigned TRUE, and therefore $\phi$ is satisfied.

\end{proof}


\section{Trees}
\label{sec:trees}
Let $P_m$ denote a path of length $m$, \emph{i.e.} a graph on vertex set $1,\ldots,m+1$ with edges $(i,i+1)$ for $1 \leq i \leq m$. Let $T$ denote a tree, \emph{i.e.} a connected graph on $n$ vertices and $n-1$ edges with no cycles.

For a valid role colouring of the vertices of a graph $G$ using $k$ colours, the \emph{role graph} $G^R$ is defined in an obvious way. $G^R$ has $k$ vertices, each one corresponding to a colour used on $V(G)$. Vertices $i$ and $j$ in $V(G^R)$ have an edge in $G^R$ if and only if vertices of colour $i$ are always connected to a vertex of colour $j$ in $G$. The graph $G^R$ may have self-loops. It is easy to see that $\Delta (G^R) \leq \Delta(G)$, $\delta (G^R) \leq \delta(G)$, where $\delta(G)$ and $\Delta(G)$ denote the minimum and maximum degree in $G$, respectively, over all vertices. If $G$ is connected, then $G^R$ must be connected. (The converse is not true.) This holds because a path in $G$ must correspond to a walk in $G^R$.

\begin{lemma} A path $P_{n-1}$ can be role coloured using $k$ colours if and only if $n=k+s(k-1)$ or $n=2k+s(2k-1)$, where $s$ is a positive integer, and {\sc Path} $k$-{\sc role colourability} is in $P$.
\end{lemma}
\begin{proof}
By the properties of $G^R$, we see that $P_{n-1}^R$ must be a path. It may have one self-loop on a leaf vertex. The path from vertex $1$ to vertex $n$ on $P_{n-1}$ corresponds to a walk on $P_{n-1}^R$ that must start and end at a vertex of degree one. If $P_{n-1}^R$ contains no self-loops then such walks can be of length $k,k+(k-1),k+2(k-1),\ldots$. If $P_{n-1}^R$ contains one self-loop then such walks can be of length $2k,2k+(2k-1),2k+2(2k-1),\ldots$. 

Checking whether either of these equalities holds is clearly in $P$, and then colouring the path is in $P$, because for a given $P_{n-1}^R$, there are at most two ways of colouring $P_{n-1}$.
\end{proof}

\begin{lemma} In a role colouring of a tree $T$, $T^R$ must be a tree with at most one self-loop. \end{lemma}

\begin{proof} Let $T^R$ be a role graph with a self-loop on vertex $w \in V(T^R)$. Consider an edge $(v_1,v_2)$ in $T$ connecting two vertices of colour $w$. If we cut the this edge, we have two components that must each contain a subgraph isomorphic to $T^R$ without the selfloop on $w$. However, if $T^R$ contains another self-loop, this process must repeat infinitely many times. 
\end{proof}

\begin{lemma} For trees, $k$-{\sc role colourability} is in $P$, if $k$ is constant or if $n-k$ is constant..
\end{lemma}
\begin{proof}
It is shown in \cite{fiala2008comparing} that, for a tree $T$, and a known role graph $T^R$ without self-loops, checking role colourability can be done in polynomial time\footnote{The algorithm given is only geared towards deciding role colourability, but it is easily transformed into an algorithm that finds an explicit colouring in polynomial time.}. By Cayley's formula, there are $k^{k-2}$ labelled trees on $k$ vertices, and $(k+1)k^{k-2}$ trees with one or no self-loops. Therefore, one can check colourability for all possible role graphs, of which there are a constant number. 
Let $T^R$ be a role graph with a self-loop on vertex $v \in V(T^R)$. Consider $T_*^R$, which is constructed as follows. Take two copies of $T^R$ without the self-loop, one with vertex set $1,\ldots,w,\ldots,k$ and a copy with vertex set $1',\ldots,w',\ldots,k'$. Let $T_*^R$ be the union of these two tree with an edge added between $w$ and $w'$. A valid $2k$-role colouring of $T$ according to $T_*^R$ now corresponds to a valid $k$-role colouring according to $T^R$ by merging the colour classes $1$ and $1'$, $2$ and $2'$, \emph{etc}.
 
\begin{claim}\label{claimrep} Let $k'=n-k$. Suppose $T$ has a valid $k$-role colouring. Let $v_1,\ldots,v_t$ be a path where $v_1$ and $v_t$ are vertices of the same colour. Then this path contains vertices of no more than $\lceil t/2 \rceil$ colours. Additionally, if $v_1$ and $v_t$ are removed from $T$, we are left with three components, $T_a,T_b,T_c$, where $T_b$ contains the path $v_1,\ldots,v_t$. Then $T_a$ and $T_c$ must contain the same colour set. \end{claim}
\begin{proof}
Vertices $v_2$ and $v_{t-1}$ must have the same colour. If not, then the path would correspond to a cycle in $T^R$, which is a contradiction. This argument can be repeated for the path $v_2,\ldots,v_{t-1}$. The neighbours of $v_1$ and $v_t$ that are not on the path $v_1,\ldots,v_t$ must have the same colour sets. 

Without loss of generality, suppose there are vertices in $T_a$ of a colour that does not appear in $T_c$. Let $v_a$ be such a vertex that is the closest to $v_1$. The second vertex on the path between $v_a$ and $v_1$ is of a colour that appears in $T_c$, but is adjacent to a vertex of a colour that does not. This is a contradiction.\end{proof} 

\begin{claim}\label{hubunique} If $T$ is $k$-role coloured with $k'=n-k$ a constant, then a vertex which cuts $T$ into more than one component of size $>2k'+1$ must have a unique colour.
\end{claim}
\begin{proof} Follows directly from claim \ref{claimrep}.\end{proof}

We define a \emph{gadget} as follows. A gadget is a maximal subtree of $T$ of size at most $2k'+1$ such that the complement of the gadget in $T$ is connected. The vertices that are adjacent to a gadget but are not in a gadget themselves are called \emph{hubs}. Gadgets are clearly non-overlapping, and hub may have many gadgets connected to it. These definitions are illustrated in figure \ref{fig:hubs}.

\begin{claim} Repeating colours can only appear within gadgets and two vertices of the same colour are either in the same gadget or in gadgets adjacent to the same hub. \end{claim}

\begin{proof}This follows from claims \ref{claimrep} and \ref{hubunique}.\end{proof}
Colour each hub with a unique colour. For every hub, go through each gadget and record all possible colourings of the subgraphs induced on the gadget and the hub, and corresponding gadget role graph (which includes a node for the unique colour of the hub). For any gadget this can be done in constant time. Additionally, keep track of all combinations of multiple gadgets that can be role coloured using the same gadget role graph. This can all be done by brute force in polynomial time, as there are $O(n)$ gadgets, $O(1)$ different gadget role graphs. Now, record a list of all the possible numbers of duplicate-coloured vertices within the hub and its gadgets. Suppose the gadgets are coloured one by one, such that those with duplicate colours are coloured first. There are $O(n-k)=O(1)$ such gadgets with a constant number of possible role colourings each. After these gadgets have been role coloured the other gadgets must be rainbow coloured. Therefore, recording all possible role colourings of the hub and its gadgets that yield no more than $k'$ duplicate-coloured vertices takes $O(n^{k'})$ time. 

Once all possible colourings for the individual hubs and their gadgets have been recorded, consider all combinations of different hub and gadget colourings. Each hub and its gadgets under one of these colourings contains $\leq k'$ duplicate-coloured vertices, and we have found a successful colouring if there is a set of hub and colouring combinations that add up to $k'$ duplicate-coloured vertices. We need $O(1)$ hubs in such a combination so, again, a brute force search of all combinations is sufficient. If such a combination does not exist, then a $k$-role colouring does not exist. Otherwise, colour the relevant hubs and their gadgets successfully, and rainbow-colour the remaining uncoloured vertices in $T$. This results in a valid $k$-role colouring of $T$. 

\begin{figure}
\begin{center}
\begin{tikzpicture}
\draw[black!20,line width=20pt] (0:0) -- (-120:1);
\draw[fill=black!20,black!20] (-120:1) circle (10pt);
\draw[fill=black!20,black!20] (2,1) circle (10pt);
\draw[fill=black!20,black!20] (5.4,1.4) circle (10pt);
\draw[fill=black!20,black!20] (4.7,-.7) circle (10pt);
\draw[black!20,line width=20pt] (0:0) -- (120:1);
\draw[black!20,line width=20pt] (120:1) -- (90:1.73);
\draw[black!20,line width=20pt] (120:1) -- (150:1.73);
\draw[black!20,line width=20pt] (2,0) -- (2,1);
\draw[black!20,line width=20pt] (4,0) -- (5.4,1.4);
\draw[black!20,line width=20pt] (4,0) -- (4.7,-.7);
\draw[fill=black!20,black!20] (90:1.73) circle (10pt);
\draw[fill=black!20,black!20] (150:1.73) circle (10pt);
\draw[fill=black!20,black!20] (-120:1) circle (10pt);
\draw[fill=white!30,white!30] (0:0) circle (10pt);
\draw[fill=white!30,white!30] (0:2) circle (10pt);
\draw[fill=white!30,white!30] (4,0) circle (10pt);
\draw[black!50,line width=.8pt] (0:0) -- (0:1);
\draw[black!50,line width=.8pt] (0:0) -- (-120:1);
\draw[black!50,line width=.8pt] (0:0) -- (120:1);
\draw[black!50,line width=.8pt] (0:1) -- (0:2);
\draw[black!50,line width=.8pt] (120:1) -- (90:1.73);
\draw[black!50,line width=.8pt] (120:1) -- (150:1.73);
\draw[black!50,line width=.8pt] (2,0) -- (2,1);
\draw[black!50,line width=.8pt] (2,0) -- (4,0);
\draw[black!50,line width=.8pt] (4,0) -- (5.4,1.4);
\draw[black!50,line width=.8pt] (4,0) -- (4.7,-.7);
\filldraw[fill=white, draw=black,line width=1pt] (0:0) circle (.1);
\draw[fill=black!100,black!100] (-120:1) circle (.1);
\draw[fill=black!100,black!100] (0:1) circle (.1);
\draw[fill=black!100,black!100] (120:1) circle (.1);
\filldraw[fill=white, draw=black,line width=1pt] (0:2) circle (.1);
\draw[fill=black!100,black!100] (2,1) circle (.1);
\draw[fill=black!100,black!100] (3,0) circle (.1);
\filldraw[fill=white, draw=black,line width=1pt] (4,0) circle (.1);
\draw[fill=black!100,black!100] (90:1.73) circle (.1);
\draw[fill=black!100,black!100] (150:1.73) circle (.1);
\draw[fill=black!100,black!100] (5.4,1.4) circle (.1);
\draw[fill=black!100,black!100] (4.7,-.7) circle (.1);
\draw[fill=black!100,black!100] (4.7,.7) circle (.1);

\end{tikzpicture}
\caption{Hubs are the vertices that separate gadgets from the rest of the tree. Here $k'=1$, hubs are white-filled and gadgets are shaded in grey.}\label{fig:hubs}
\end{center}
\end{figure}
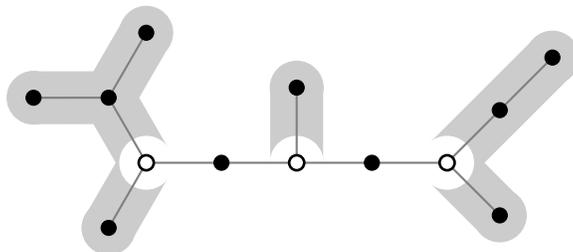

\end{proof}

\FloatBarrier

\section{Cographs}
\label{sec:cographs}

\emph{Cographs} are exactly the $P_4$ free graphs~\cite{seinsche1974property}. The join of two graphs $G_1$ and $G_2$ is the graph $G_3=G_1+G_2$ such that $V(G_3)=V(G_1) \cup V(G_2)$ and $E(G_3)=E(G_1) \cup E(G_2) \cup \{ (i,j) | i \in V(G_1), j \in V(G_2) \} $. The disjoint union of two graphs $G_1$ and $G_2$ is the graph $G_3=G_1 \cup G_2$ such that $V(G_3)=V(G_1) \cup V(G_2)$ and $E(G_3)=E(G_1) \cup E(G_2)$. Cographs can be constructed recursively from $K_1$ by disjoint union and join. They are the smallest class of graphs closed under disjoint union and join. Every cograph $G$ has an associated (binary, not necessarily unique) \emph{cotree}, whose leaves correspond to the vertices of $G$, and the non-leaves are labelled ``0" and ``1" to denote disjoint unions and joins, respectively. The cotree describes how $G$ is formed from instances of $K_1$ by successive joins and disjoint unions. Given a cograph, its cotree can be found in linear time~\cite{cournier1994new,mcconnell1994linear}.

\begin{theorem}
All cographs with $\geq 2$ vertices are 2-role-colourable and 2-{\sc rolecol} for cographs is in P.
\end{theorem}

\begin{proof}
Suppose $G$ is a cograph. If $G$ is not connected, then we can mono-colour each component either red or blue, such that both red and blue are used and such that, if $G$ contains both isolated vertices and components of size $\geq 2$, these two types of components receive different colours.

Therefore, suppose $G$ is connected. A connected cograph $G$ with $|V(G)|=2$ is isomorphic to $K_2$, and can both be 2-role-coloured in the obvious way. Suppose that all connected cographs $G'$ with $|V(G')|<k$ can be 2-role-coloured. Suppose $|V(G)|=k>2$ and the last step in the construction of $G$ was a join of graphs $G_1$ and $G_2$ (which it must be if $G$ is connected). We consider three separate cases.
\begin{description}
\item[(i)] If $|V(G_1)|,|V(G_2)|>1$, then we can 2-role-colour the vertices of $G_1$ (red and blue) and $G_2$ (red and blue). This extends to a valid 2-role-colouring of $G$, because all vertices have red and blue neighbours. 

\item[(ii)] If $|V(G_1)|=1$, $|V(G_2)|>1$ (without loss of generality) and $G_2$ is 1-role-colourable, then colour the vertex of $G_1$ red and the vertices of $G_2$ blue. This extends to a valid 2-role-colouring of $G$, because the red vertex has only blue neighbours and the blue vertices have only a red neighbour or red and blue neighbours, depending on whether $G_2$ is empty or not.

\item[(iii)] If $|V(G_1)|=1$, $|V(G_2)|>1$ (without loss of generality) and $G_2$ is not 1-role-colourable, then colour the vertex of $G_1$ red. The graph $G_2$ is not 1-role-colourable, which means that is is disconnected with isolated vertices and components with $\geq 2$ vertices. Colour the isolated vertices of $G_2$ blue. For each component of $G_2$ with $\geq 2$ vertices, colour one vertex blue and the others red. This extends to a valid 2-role colouring of $G$, because all blue vertices have only red neighbours and all red vertices have both red and blue neighbours.
\end{description}

Therefore, we can always find a valid 2-role-colouring for $G$. It is easy to see that this method can be executed in polynomial time. We can find a cotree in polynomial time, which gives us a $G_1$ and $G_2$. Then we find the connected components of $G_1$ and $G_2$, which can also be done in polynomial time (by a series of at most $n$ breadth first searches). 
\end{proof}

\begin{theorem}
All cographs with $\geq k$ vertices are $k$-role-colourable and $k$-{\sc rolecol} for cographs is in P, where $k>2$.
\end{theorem}

\begin{proof}
We know that all cographs with $\geq 2$ vertices are 2-role-colourable, so suppose that all cographs are $k'$-role-colourable for all $2 \leq k' <k$. Suppose $G$ is a cograph with $|V(G)|=n\geq k$ and the last step in a construction of $G$ was either a join or a disjoint union of $G_1$ and $G_2$, with $|V(G_1)|=n_1$ and $|V(G_2)|=n_2$. Pick $k_1$ and $k_2$ such that $k_1+k_2=k$, $k_1\leq n_1$, $k_2\leq n_2$ and $k_i=1$ only if $n_i=1$, for $i=1,2$. Now, $G_i$ is $k_i$-role-colourable by our inductive assumption, for $i=1,2$. Note that $k_i$ is only equal to 1 if $n_i$ is 1, and $K_1$ is always 1-role-colourable. So, we can colour $G_1$ using $k_1$ colours and $G_2$ using $k_2$ different colours. This extends to a valid $k$-role-colouring of $G$ regardless of whether $G=G_1 \cup G_2$ or $G=G_1 + G_2$.
\end{proof}

\section{Acknowledgements}
Puck Rombach is supported by AFOSR MURI grant FA9550-10-1-0569 and ONR grant N000141210040. Part of this work was undertaken while Puck Rombach was attending the	semester	program 	``Network	Science	and	Graph	Algorithms"	at	the	Institute	for	
Computational	and	Experimental	Research	in	Mathematics	(ICERM) at Brown University.
\bibliographystyle{plain}
\bibliography{role}

\end{document}